\documentclass{svmult} 

\usepackage{graphicx} 
\usepackage{multicol}
\usepackage[bottom]{footmisc}

\usepackage[utf8]{inputenc} 

\usepackage{mathtools}
\usepackage{hyperref}
\hypersetup{
    unicode=false,          
    pdfkeywords={mesh} {smoothing} {global optimization} {efficient}, 
    colorlinks=true,       
    pdfborder=0 0 0,
    linkcolor=blue,          
    citecolor=blue,        
    filecolor=blue,      
    urlcolor=blue           
}

\usepackage{xcolor}
\usepackage{lineno}

\usepackage{array} 
\usepackage{paralist} 

\usepackage{fancyhdr} 
\newcommand{\cg}{\color{gray}}

\newcommand{\R}{\mathbf{R}}

\newcommand{\iq}{\operatorname{iq}}

\newcommand{\diag}{\operatorname{diag}}
\newcommand{\tr}{\operatorname{tr}}

\newcommand{\vol}{\operatorname{vol}}
\newcommand{\area}{\operatorname{area}}
\def\co{\colon\thinspace}
\def\la{\langle}
\def\ra{\rangle}
\newcommand{\rqedhere}{\tag*{\hspace{-1em}\qed}}

\begin{document}
\title*{Efficient and Global Optimization-Based
Smoothing Methods for Mixed-Volume Meshes}
\titlerunning{An Efficient and Global Optimization-Based Smoothing Method}

\author{Dimitris Vartziotis\inst{1,2,3} \texorpdfstring{and}{\and} Benjamin Himpel\inst{1}}

\institute{
  \it TWT GmbH Science \& Innovation, Department for Mathematical Research \& Services,
  Bernh{\"a}user Stra{\ss}e~40--42,
  73765~Neuhausen, Germany \and
  NIKI Ltd.\ Digital Engineering, Research Center,
 205~Ethnikis Antistasis Street, 45500~Katsika, Ioannina, Greece \and
 Corresponding author. E-mail address: dimitris.vartziotis@nikitec.gr
}



\linenumbers
\maketitle
\begin{abstract}Some methods based on simple regularizing geometric element transformations have heuristically been shown to give runtime efficient and quality effective smoothing algorithms for meshes. We describe the mathematical framework and a systematic approach to global optimization-based versions of such methods for mixed volume meshes. In particular, we identify efficient smoothing algorithms for certain algebraic mesh quality measures. We also provide explicit constructions of potentially useful smoothing algorithms.
\end{abstract}




\section{Introduction}

In the context of the finite element method mesh quality affects numerical stability as well as solution accuracy of this method \cite{StrangFix2008}. The class of geometric element transformation methods (GETMe) consists of mesh smoothing methods based on simple geometric element transformations. In \cite{VartziotisWipperSchwald2009} such a smoothing algorithm is introduced for tetrahedral meshes based on shifting vertices of a tetrahedron by the opposing face normals, normalized to be scaling-invariant. It has been tested numerically and extended to other volume types in a series of papers \cite{VartziotisWipperSchwald2009,VartziotisWipper2011Hex, VartziotisWipper2011Mixed, VartziotisPapadrakakis2013,VartziotisWipperPapadrakakis2013}.

In \cite{VartziotisHimpel2013} we show that the mean volume can be viewed as a quality measure for tetrahedron, pyramid, prism and hexahedron, and that the discretization of its gradient flow is a natural generalization of the aforementioned tetrahedral GETMe algorithm. It enables us to prove that it regularizes certain polyhedron types and therefore is a locally optimization-based method. For tetrahedra, this quality measure is related to the mean ratio quality measure \cite{Knupp2001}. Note that our generalization to other polyhedron types is inherently very different from previous extensions via dual polyhedra to hexahedra, prisms and pyramids \cite{VartziotisWipper2011Hex,VartziotisWipper2011Mixed} as well as the mean ratio criterium for other polyhedra \cite{Knupp2001}. 

Several numerical tests have shown, that the GETMe approach combines the advantages of Laplacian
smoothing variants in terms of runtime efficiency with the mesh
quality effectiveness of a global optimization-based smoothing \cite{VartziotisWipperPapadrakakis2013}. While we have shown in \cite{VartziotisHimpel2013} that our generalizations of the geometric element transformations in \cite{VartziotisWipperSchwald2009} regularize certain polyhedron types, we have not touched upon the subject of global optimization for meshes, which is of particular interest to the meshing community. The main goal of our work is to describe the mathematical framework, which allows us to systematically construct global optimization-based GETMe smoothing methods reminiscent of the transformation in \cite{VartziotisWipperSchwald2009}.  As a result we identify some simple geometric element transformations, which optimize certain algebraic mesh quality measures. This approach allows for rigorous analysis, natural generalizations and problem-specific extensions. We hope, that it will also serve as a basis for new runtime efficient and quality effective smoothing algorithms.

In Sect.\ \ref{sec_framework} we introduce the notions of mesh quality and show, how we can turn the mean volume function into a useful a volume element quality measure, which is equivalent to the mean ratio criterium for tetrahedra. Sect.\ \ref{sec_getme}  describes how the GETMe algorithm in \cite{VartziotisWipperSchwald2009} fits into our mathematical framework. In Sect.\ \ref{sec_optimization_approach} we see, how the mean volume function gives the simple geometric element transformation underlying the tetrahedral GETMe algorithm, and construct other simple geometric element transformations based on quality measures. Most importantly, Theorem \ref{thm_main} shows, when homogeneous degree $d$ vector fields give global optimization-based scaling-invariant smoothing methods. In particular, this gives a way to construct simple geometric element transformations optimizing certain algebraic mesh quality measures. In Sect.\ \ref{sec_generalizations} we generalize our results in various ways. We give concrete 
formulas generalizing the tetrahedral 
geometric element transformation in \cite{VartziotisWipperSchwald2009} to the pyramid, the prism and the hexahedron, and briefly discuss 
the generalization to meshes in arbitrary dimensions, the isoperimetric quotient as a volume element quality measure, whose discretized gradient flow gives yet another GETMe smoothing method, and shape preservation.

\section{The Mathematical Framework}\label{sec_framework}

It is not apparent why a simple geometric transformation for tetrahedra given by shifting vertices by the opposing face normals should yield such a reliable and efficient smoothing algorithm for meshes. An elegant mathematical explanation provides a way to systematically develop and analyze some of these promising GETMe algorithms further. This is the motivation for creating a suitable mathematical framework, which was used in the proof of the regularizing behavior for the transformation in \cite{VartziotisWipperSchwald2009}. Alas, instead of showing that the core geometric transformation in \cite{VartziotisWipperSchwald2009} is a useful smoothing approach, we find other related GETMe algorithms, which might be more promising.

A global optimization-based smoothing method increases in each iteration step the value of a (global) quality function for meshes by repositioning the vertices of a mesh. A word of caution up front for mathematicians: local and global optimization refers to whether a mesh is optimized locally in a neighborhood of a vertex or globally for the entire mesh. In particular, global optimization generally refers to the search for a local optimum of a global quality function. Global optimization-based smoothing methods are known to yield meshes with superior element quality \cite{DiachinKnuppMunsonShontz2006,BrewerDiachinKnuppLeurentMelander2003}. It is a common theme in mathematics, in particular in topology and geometry, that strong mathematical results are based on optimizing functions. The heat equation is one of the most important evolutionary differential equations and can be interpreted as the gradient flow of the Dirichlet energy functional. Perelman observed in his proof of Thurston's 
geometrization conjecture \cite{KleinerLott2008_NotesPereleman}, that one can interpret Hamilton's Ricci flow as a gradient flow. This is also the reason for the favorable properties of discrete Ricci flow introduced in \cite{ChowLuo2003_CombinatorialRicci}, which has successfully been used in combination with conformal geometry to optimize surface parametrization in computer graphics \cite{GuYau2008_ComputationalConformal}. Morse theory \cite{Morse1996_CalculusVariations,Bott1982_LecturesMorse} with all its generalizations and infinite-dimensional manifestations in mathematical physics like Chern--Simons theory \cite{Witten1989_QuantumField}, Yang--Mills theory \cite{AtiyahBott1983_YangMills,Atiyah1988_CollectedWorksGaugeTheories,Rubakov2002_ClassicalTheory} and other gauge theories deduces topological information from the study of the singularities and the flow of a gradient field. Our goal is to find simple geometric element transformations, which give global optimization-based GETMe algorithms, in order 
to 
combine the mathematical advantages of global optimization with the speed of GETMe. In order to be successful, we need to work in a good mathematical framework.

In order to keep the mathematical overhead to a minimum, we adapt the notation from \cite{DiachinKnuppMunsonShontz2006} and refer to \cite{VartziotisHimpel2013} for a complete treatment offered to the interested reader.

\subsection{Element and Mesh Quality}\label{sec_quality}

A mesh consists of an ordered set of $n$ vertices $V$ and $|E|$ volume elements $E$. Let us write $V$ as a tuple $(v_1,\ldots, v_n)$. Let $x_v \in \R^3$ denote the coordinates for the vertex $v \in V$, so that $x = (x_{v_1},\ldots, x_{v_n}) \in \R^{3\times n}$ is the tuple of all vertex
coordinates. Each volume element $e \in E$ consists of an ordered set $V_e$ of $n_e$ vertices of $V$ and the edges between these vertices. If the edges are fixed, then we will simply identify $e = V_e$.  Note that the order of vertices determines a preferred orientation for $e$. Let $x_e \in \R^{3\times n_e}$ be the collection of coordinates for e.
Associated with each element $e$ is an element quality measure
$q_e\co\R^{3\times n_e} \to \R$ as a continuous function of the vertex positions, where a larger value of $q_e$ indicates a higher quality element. The overall quality of a mesh is measured by a continuous function $Q\co \R^{|E|} \to \R$, taking as input the vector of volume element qualities and combining them in a certain way, for example by taking the arithmetic mean, but even combinations depending on the size, location, direction, and other properties are conceivable.

In order to simplify the notation, we call the function
\begin{align*}
q\co \R^{3\times n} &\to \R\;,\\                                                                                                                                                                                                                                                                                                                                                                                                                        
x&\mapsto Q\left(\prod_{e\in E} q_e(x_e)\right)
\end{align*}
the mesh quality measure, where $\prod$ is the Cartesian product. In general, useful quality
measures possess other properties in addition to continuity, like invariance under translation, scaling, rotation and reflection \cite{Knupp2001}. Therefore, let us consider only mesh quality measures, which are invariant under translation and scaling, even though our discussion would be a lot less technical without this assumption. Alternatively, we can consider the coordinates $x \in \R^{3\times n}$ up to scaling and translation. If we disregard the one-point meshes, where all vertex coordinates are equal, the resulting space is simply a $(3n-4)$--dimensional sphere $S^{3n-4} \subset \R^{3n-3}$. We can describe this space concisely as a quotient of a subspace of $\R^{3\times n}$ by the equivalence relation given by translation and scaling \cite[Sect.\ 2]{VartziotisHimpel2013}. This observation allows us to view a translation- and scaling-invariant quality measure $q \co \R^{3\times n} \to \R$ as a function on this $(3n-4)$--dimensional sphere embedded in $\R^{3 \times n}$. Since spheres are compact 
manifolds without boundary and a global optimization-based method is essentially 
the gradient flow of a function on this sphere, we can now unleash the power of geometric analysis to study this method.

By the above argument, we have not only gained a different angle from which we can view and study quality measures, but an entirely new approach to constructing optimization-based quality measures. In fact, any function on a $(3n-4)$--sphere embedded in $\R^{3\times n}$ representing all meshes up to translation and scaling gives rise to a translation- and scaling-invariant quality measure, whose gradient flow will give a global optimization-based smoothing method through the method of steepest descent.

\subsection{The Mean Volume Function}\label{sec_mean_vol}

The signed volume of a tetrahedron with vertex coordinates $x=(x_1,\ldots,x_4) \in \R^{3\times 4}$ is given by
\begin{equation}\label{eq_vol}
\vol(x) = \frac{1}{6}((x_2-x_1)\times (x_3-x_1))\cdot (x_4-x_1)\;.
\end{equation}
The orientation of the tetrahedron and therefore the sign of the volume function is determined by the order of vertices. Notice, that it is a well-defined function for all meshes with the structure of a tetrahedron. This function naturally extends to convex polyhedra by first triangulating them, in particular to hexahedra, prisms and pyramids. By averaging over all possible triangulations, we can extend the volume function to a mean volume function on all meshes with the edge structure of a convex polyhedron and again denote it by $\vol$.

There are alternative definitions for this mean volume function, but this is the most convenient one for our purpose. For example, in the case of a pyramid mesh $e=(v_1,\ldots,v_5)$, where $v_5$ is the apex, we only have the two different triangulations depicted in Fig.\ \ref{pyramid_triang}. For a detailed discussion of triangulations and the mean volume function see \cite[Sect.\ 5]{VartziotisHimpel2013}.
\begin{figure}[ht]
\def\svgwidth{7cm}
\centering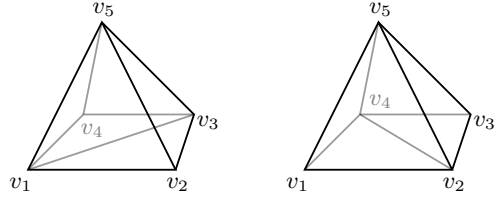\caption{The two different triangulations of a pyramid\label{pyramid_triang}}
\end{figure}

Clearly, the mean volume function is translation-invariant, but not scaling-invariant. However, if we consider its restriction to a $(3n-4)$--sphere representing all meshes with the structure of a fixed convex polyhedron with $n$ vertices up to translation and scaling, this yields a scaling- and translation-invariant quality measure $q_e$ for polyhedra, which is essentially maximized by regular polyhedra \cite{VartziotisHimpel2013}. To be more concrete, the quotient map by translation and scaling viewed as a projection to the submanifold
\[
 N\coloneqq\{x \in M \mid \|x\|=1 \text{ and } x_*\coloneqq   \sum_{i=1}^n x_i = \vec{0}\} \subset M
\]
is given by
\begin{align*}
\pi\co \R^{3\times n}\backslash\{(x_0,\ldots,x_0)\mid x_0\in \R^3\} & \to N \;,\\ 
(x_1,\ldots,x_n) & \mapsto \frac{(x_1-x_*,\ldots,x_{n}-x_*)}{\|(x_1-x_{*},\ldots,x_{n}-x_*)\|}\;,
\end{align*}
and the diffeomorphism from $N$ to $S^{3n-4} \subset \R^{3(n-1)} \cong \R^{3\times (n-1)}$ is given by
\begin{align*}
(x_1,\ldots,x_n) & \mapsto \frac{1}{\|(x_1,\ldots,x_{n-1})\|}(x_1,\ldots,x_{n-1})\;.
\end{align*}
Similarly, let $N_e \subset \R^{n_e}$ be the sphere corresponding to $e$.

\subsection{The Mean Ratio Metric}

Algebraic mesh quality measures had been introduced by \cite{Knupp2001}, and have been influential for smoothing, optimization and edge swapping techniques. These quality measures also fit nicely in our framework. Let us consider the mean ratio quality measure. It is given by
\begin{equation}
  \label{eq_meanratio}
  q(x_e):=\frac{3\det (S)^{2/3}}{\|\,S\|_F^2}\,,
\end{equation}
with $\|\,S\|_F:=\sqrt{\tr(S^{\mathrm{t}}S)}$ denoting the
Frobenius norm of the matrix $S:=DW^{-1}$. Here $D$ represents the
difference matrix given by 
\begin{equation}
  \label{eq:diffmatrix}
   D:=(x_2-x_1,x_3-x_1,x_4-x_1) \quad \text{for }\det(x_e)>0
\end{equation}
and $W$ denotes the difference matrix of a reference tetrahedron. It holds
that $q(x_e)\in[0,1]$, where very small values indicate nearly degenerated
elements and larger values elements of good quality. In particular it holds
that $q(x_e)=1$, if $x_e$ is regular.

When the desired shape is a regular tetrahedron, the Frobenius norm is simply a way to make the volume scaling-invariant while preserving the translation- and rotational invariance. The main reason for this choice is that this matrix norm is easily implemented and efficient, and has therefore been used in quality measures. Now, $\det(S) = \det(D)\det(W^{-1}) = \vol(x_e)/w$, where $w = \det(W)$ is constant. Furthermore $\det(D) = 6 \vol(x_e)$, when $e$ is a tetrahedron. We can rewrite
\[
 C \left(q(x_e)\right)^{3/2} =  \vol\left(\frac{1}{\|S\|_F}x_e\right) \quad \text{for some constant } C = C(W)> 0.
\]
The volume function restricted to the submanifold of $M$ diffeomorphic to a $(3n-4)$--sphere given by
\[
\{p \in M \mid x_*= \vec{0} \text{ and } x_e = \|S\|_F\} \subset M
\]
is therefore an equivalent way of describing the mean ratio quality measure.

With the results from \cite{VartziotisHimpel2013} in hand, we see, that \cite{VartziotisWipperSchwald2009} starts with the same volume function, but normalizes its gradient rather than the quality measure, so that it is scaling-invariant. Furthermore, we have proven, that this scaling-invariant gradient optimizes the quality measure given by a certain normalized volume. In effect, the aim of both methods is the optimization with respect to the volume function. This is a heuristic explanation, why Mesquite and GETMe give meshes of similar quality.

\section{GETMe}\label{sec_getme}

The original GETMe algorithm for tetrahedral meshes \cite{VartziotisWipperSchwald2009} is based on shifting vertices by opposing face normals. Let us give a few details in order to describe how this algorithm fits into the mathematical framework. 

\subsection{Scaling-Invariance}

For a tetrahedral volume element $e=(v_1,\ldots,v_4)$ consider the vector
\[
 X_e = \begin{bmatrix}
        X_{e,1}\\
        X_{e,2}\\
        X_{e,3}\\
        X_{e,4}
       \end{bmatrix}
       \coloneqq
\begin{bmatrix}
(x_{4} -x_{3}) \times (x_{3}-x_{2})\\
(x_{4} -x_{1}) \times (x_{1}-x_{3})\\
(x_{4} -x_{2}) \times (x_{2}-x_{1})\\
(x_{1} -x_{2}) \times (x_{2}-x_{3})\\
\end{bmatrix} \in \R^{3\cdot 4}\cong \R^{3\times 4}\;,
\]
where we identify the column vector of (column) vector entries with the row vector of the same vector entries. The GETMe algorithm without the imposed scaling-invariance applied only to this one volume element $e$ transforms its coordinates via
\begin{align*}
  \R^{3\times 4} &\to \R^{3\times 4}\;,\\
 x_{e} &\mapsto x'_{e}= x_e + \sigma X_e
\end{align*}
for some parameter $\sigma >0$ independent of $e$. The tetrahedral GETMe smoothing introduced in \cite{VartziotisWipperSchwald2009} applied to an entire mesh then shifts $x_i$ for $i=1,\ldots,n$ by 
averaging all (vector) coordinates of $X_e$ over all $e$ containing $v_i$.  More precisely, if $V = (v_1,\ldots,v_n)$ are the vertices of a mesh with coordinates $x$ and $e = (w_1,\ldots,w_4)$ is a tetrahedral volume element within the mesh, then let us introduce the notation
\begin{align*}
  \phi_e\co \quad \R^{3 \times 4} & \hookrightarrow \R^{n\times 4}\;,\\
  (y_1,\ldots,y_4) & \mapsto (x_1,\ldots,x_n), \text{ where}\begin{cases}
                      x_i=y_j & \text{if } v_i = w_j\;,\\
                      x_i=\vec{0} & \text{else}\;,
                     \end{cases}
\end{align*}
and let $\lambda_i$, $i=1,\ldots,n$ be the number of volume elements containing $v_i$. If we set
\[
 X \coloneqq  \diag(\lambda_1^{-1},\ldots,\lambda_n^{-1})\sum_{e \in E}\phi_e(X_e)\;,
\]
then the original GETMe transformation is given by
\begin{align*}
  \R^{3\times n} &\to \R^{3\times n}\;,\\
x &\mapsto x' = x + \sigma X\;.
\end{align*}
However, in order to make it scaling-invariant, we need to modify $X$. We have arrived at a subtle issue, where our  solution in \cite{VartziotisHimpel2013} and the one in \cite{VartziotisWipperSchwald2009} differ, and which we will encounter again in Sect.\ \ref{sec_scaling}. This can only be appreciated when considering GETMe from a mathematical rather than a heuristic point of view. In praxis it does not make much of a difference, how we make the algorithm scaling-invariant, but there is a preferred way within our mathematical framework. In \cite{VartziotisWipperSchwald2009} each face normal in $X_{e,j}$ was divided by the square root of its norm $\|X_{e,j}\|$, which seemed to work very well for all practical purposes. In the language of our mathematical framework on the other hand, the vector \[X\in \R^{3\times n}\] determines the dynamical behavior on the $(3n-4)$--sphere rather than the vectors $X_{e,j} \in \R^3$ individually, which is why we would prefer not to change the direction of $X$. Therefore,
 we divide $X$ by the square root of its norm $\|X\|$, which results in 
a scaling of $X_e$ depending on the size of the volume element $e$.
In summary, if we define the normalization function
\begin{align*}
 \Psi\co \R^{3\times n} &\to \R^{3\times n}\;,\\
 X &\to \begin{cases}\frac{1}{\sqrt{\|X\|}} X & \text{if } X\neq \vec 0\\
         0 & \text{if } X=\vec 0\;,
        \end{cases}
\end{align*}
then we consider the transformation $T_\sigma$ where
\begin{align*}
T_\sigma\co \R^{3\times n} &\to \R^{3\times n}\;,\\
 x&\mapsto x' = x + \sigma \, \Psi(X)
 \end{align*}
as our variation of the original GETMe smoothing algorithm \cite{VartziotisWipperSchwald2009}.

Let us dwell a little longer on this issue. If we take a closer look, we choose a scaling different from \cite{VartziotisWipperSchwald2009} twice:
\begin{enumerate}
 \item in the normalization for each individual volume element, and
 \item in the averaging procedure over all volume elements for the mesh.
\end{enumerate}
From a global point of view we are closer to the initially conceived GETMe algorithm for tetrahedral meshes without the imposed scaling-invariance than we have been in \cite{VartziotisWipperSchwald2009}, because the vector $\Psi(X)$ points in the same direction as $X$. As we will see in Sect.\ \ref{sec_concrete_quality}, this entirely rigorous approach to the algorithm presented \cite{VartziotisWipperSchwald2009} is not very useful, because the latter includes other scaling methods and additional control mechanisms in order to exhibit such favorable behavior.

\subsection{Face Normals}

As we have seen, face normals play an essential role in the tetrahedral GETMe approach. By the face normal of an oriented triangle we simply mean the cross product of two edge vector compatible with the orientation. In order to simplify notation, let us define face normals for arbitrary polygonal curves in $\R^3$ by
\[
\nu(1, \ldots, k) \coloneqq x_{1} \times x_{2} + x_{2} \times x_{3} + \ldots + x_{{k-1}} \times x_{k} + x_{k} \times x_1\;.
\]
Clearly, if $(x_1,x_2,x_3)$ are the coordinates of a triangle in $\R^3$, then \[\nu(1,2,3) = (x_{1} -x_{2}) \times (x_{2}-x_{3})\] is the usual face normal considered in the previous section. In particular, we can write
\begin{equation}\label{eq_getme_tetra}
 X_e = \begin{bmatrix}
\nu(4,3,2)\\
\nu(4,1,3)\\
\nu(4,2,1)\\
\nu(1,2,3)
\end{bmatrix}
\text{ for a tetrahedral element }e=\text{\raisebox{-1cm}{\def\svgwidth{2.5cm}
\begingroup%
  \makeatletter%
  \providecommand\color[2][]{%
    \errmessage{(Inkscape) Color is used for the text in Inkscape, but the package 'color.sty' is not loaded}%
    \renewcommand\color[2][]{}%
  }%
  \providecommand\transparent[1]{%
    \errmessage{(Inkscape) Transparency is used (non-zero) for the text in Inkscape, but the package 'transparent.sty' is not loaded}%
    \renewcommand\transparent[1]{}%
  }%
  \providecommand\rotatebox[2]{#2}%
  \ifx\svgwidth\undefined%
    \setlength{\unitlength}{126.71829834bp}%
    \ifx\svgscale\undefined%
      \relax%
    \else%
      \setlength{\unitlength}{\unitlength * \real{\svgscale}}%
    \fi%
  \else%
    \setlength{\unitlength}{\svgwidth}%
  \fi%
  \global\let\svgwidth\undefined%
  \global\let\svgscale\undefined%
  \makeatother%
  \begin{picture}(1,0.87551156)%
    \put(0,0){\includegraphics[width=\unitlength]{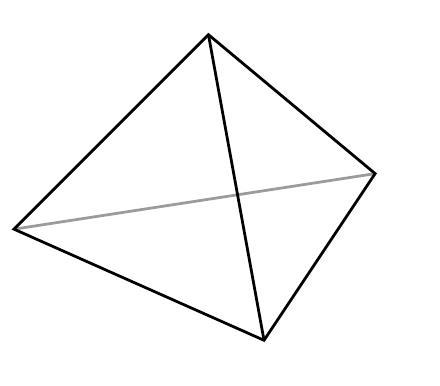}}%
    \put(-0.00419237,0.26444029){\color[rgb]{0,0,0}\makebox(0,0)[lb]{\smash{$v_1$}}}%
    \put(0.5414497,0.01191126){\color[rgb]{0,0,0}\makebox(0,0)[lb]{\smash{$v_2$}}}%
    \put(0.8526012,0.41776081){\color[rgb]{0,0,0}\makebox(0,0)[lb]{\smash{$v_3$}}}%
    \put(0.39714778,0.83713904){\color[rgb]{0,0,0}\makebox(0,0)[lb]{\smash{$v_4$}}}%
  \end{picture}%
\endgroup%
}.}
\end{equation}
For polygonal curves in a hyperplane of $\R^3$, $\nu$ is a natural generalization of the face normal from triangles to the planar surface enclosed by the curve. The direction of $\nu$ can be determined by the right-hand (grip) rule. For arbitrary polygonal curves in $\R^3$, $\nu$ corresponds to the sum of the face normals of an arbitrary triangulated surface, whose boundary is the curve. For more details on face normals, please consult \cite[Sect.\ 3]{VartziotisHimpel2013}.

\section{An Efficient and Optimization-Based Approach}\label{sec_optimization_approach}

Often, quality effective smoothing methods are developed using global optimization. In fact, global optimization-based methods yield meshes of superior quality. It is rather fortunate, that we are able to find our way back from GETMe to the basics of global optimization via reverse engineering, because the resulting algorithms will combine the best of both worlds, being not only very efficient, but also effective. We will also discuss how previously tested algorithms fit into this framework.

\subsection{Suitable Mesh Quality Functions}\label{sec_concrete_quality}

As we have mentioned in Sect.\ \ref{sec_quality}, any function on $N$ can be viewed and used as a global mesh quality function, whose gradient flow will find local maximums. There are lots of different ways of combining element quality functions in order to get a global one for meshes. Our immediate goal is to find potentially useful measures, whose gradients preserve the GETMe characteristic.

The most naive option is to add all the element quality functions up. A straight-forward computation as in \cite[Sect.\ 3]{VartziotisHimpel2013} shows, that the mean volume for meshes defined in Sect.\ \ref{sec_mean_vol} satisfies
\[
  6 \nabla \vol = X\;,
\]
where $\nabla \vol$ is the gradient vector field of $\vol$. The method of steepest descent applied to six times the negative mean volume clearly yields the transformation \[x' = x + \sigma X \quad \text{for }\sigma \text{ sufficiently small}.\] Up to the averaging process at each node, this is the original GETMe algorithm \cite{VartziotisWipperSchwald2009} without the imposed scaling-invariance, relaxation, weighted averages and any further control mechanisms. This might provide some confidence, that this would yield an efficient and effective GETMe algorithm. Instead, this uncovers a big drawback of the core GETMe transformation: Vertices, which are not on the boundary surfaces are fixed when shifting via $X$, because the mean volume is independent of the location of the inner vertices! More concretely, consider a regular tetrahedra with a single inner vertex and four smaller tetrahedra whose edges connect the outer vertices $v_1,\ldots,v_4$ with the inner vertex $v_5$ as in Fig.\ \ref{fig_fixedpoint}.
\begin{figure}[ht]
\def\svgwidth{4cm}
\centering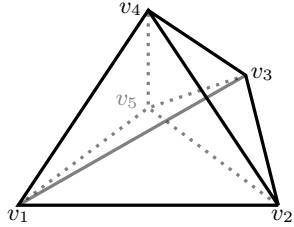\caption{The triangulation of a regular tetrahedra with one inner vertex\label{fig_fixedpoint}}
\end{figure}

The inner point is fixed by the flow of $X$ and and stays fixed after averaging the corresponding face normals, no matter where it is located, because the opposite face normals add up to the trivial vector. One could argue, that by rescaling the regular tetrahedra, the inner point moves towards the center. However, if we have more than one inner point, all of them would move towards the center. Therefore, the correct way to rescale the resulting mesh is by projecting the boundary vertices onto the original boundary surface, even though this means, that the inner vertices are fixed. It is therefore surprising, that the GETMe algorithm in \cite{VartziotisWipperSchwald2009} works as well as it does. In any case, a systematic approach might lead to a more stable and mathematically predictable algorithm without the need of control mechanisms.
 
Instead of simply using the mean volume as a global mesh quality function, the sum of the scaling-invariant element quality functions $q_e$ \[q(x) = \sum_{e\in E} q_e(x_e)\] given by first projecting $x_e$ to $N_e$ would be a more promising albeit more complicated approach, which we will not pursue at this point. However, there are other ways of combining the mean volumes to give potentially useful mesh quality functions. Let us consider the restriction of
\begin{equation}\label{eq_vol_measures}
 q_1(x) = \prod_{e\in E}\vol(x_e)^2 \quad \text{and} \quad q_2(x) = -\sum_{e\in E}\frac{1}{\vol(x_e)^2}
\end{equation}
to $N$. We compute
\[
\nabla q_1 = k_1(x)\sum_{e\in E}\frac{1}{\vol(x_e)} X_e  \quad \text{and} \quad \nabla q_2 = k_2(x)\sum_{e\in E}\frac{1}{\vol(x_e)^2} X_e\;,
\]
where
\[
k_1(x)= 2\prod_{e\in E}\vol(x_e)^2 \quad \text{and} \quad k_2(x) = 4\;.
\]
Obviously, vanishing (and negative) mean volume is problematic. For convenience, let us assume that all of our volume elements have positive mean volume. In particular, all valid volume elements should satisfy this mild assumption. For tetrahedra, positive mean volume is equivalent to them being valid. In particular, both quality measures ensure in the case of tetraehedra, that they stay valid under sufficiently small transformations using the gradient. However, note that we could define useful quality measures for all volume elements by shifting the volume to be positive. Since $N$ is compact, we could shift the volume by
\[
 c = \max_{e\in E}\max_{N_e} \vol\;.
\]

We observe, that $\nabla q_1$ and $\nabla q_2$ look similar, and therefore we expect them to behave similarly. Furthermore, in the above concrete example of the regular tetrahedron triangulated by 4 tetrahedra, the inner point is not fixed by the gradient flow. Instead, it moves away from the smaller (and more irregular) tetrahedra towards the center. In general, the resulting transformation seem to prefer volume element distributions, which are similar in size and at the same time as regular as possible. There are a lot of other conceivable mesh quality functions and combinations of them. In this paper we emphasize the mathematical framework rather than specific algorithms. However, preliminary tests in the two-dimensional analogue using the area instead of the volume show, that the gradient flow givs a powerful GETMe smoothing method.

\subsection{Scaling-Invariance}\label{sec_scaling}

Whatever the gradient flow does, it would be nice to have a scaling-invariant flow of some related vector field. This can be achieved just like we have done in \cite{VartziotisHimpel2013} by introducing an appropriate factor. The proof of the following result is entirely analogous to the proofs of \cite[Lemma 3.3 and Theorem 4.2]{VartziotisHimpel2013}, but it is more general.
\begin{theorem}\label{thm_dynamics}
 Let $\nabla_\pi$ denote the gradient with respect to the induced submanifold metric $\la\cdot,\cdot\ra_\pi$ on $N$. Let $X$ be a homogeneous vector field of degree $d$ and $q$ a quality measure satisfying
 \[
  d(\pi)(X_x) = c_x\nabla_\pi q|_x \quad \text{and } \quad c_x>0\quad \text{for }x\in N.
 \]
Consider $\Psi\co \R^{3\times n} \to \R^{3\times n}$ given by
 \[\Psi(X) = \|X\|^{(1-d)/d} X\;.\]
Then the vector fields on $N$ given by
 \begin{equation}\label{eq_vectorfields}
 Y \coloneqq D(\pi)(X) \quad \text{and}\quad \tilde Y \coloneqq D(\pi)(\Psi(X))
 \end{equation}
start and end at singularities of $Y$, and the vector field $\tilde Y$ on $N$ corresponds to $\Psi(\nabla q)$ on the quotient manifold $\R^{3\times n}\backslash\{(x_0,\ldots,x_0)\mid x_0\in \R^3\}/\sim$.
\end{theorem}

In particular, the above theorem applies to $q_i$ defined in Sect.\ \ref{sec_concrete_quality}, where $c_x=(k_i(x))^{-1}$ and $X_i = c_x\nabla q_i|_x$ is homogeneous of degree $-(3i+1)$.

In \cite[Theorem 4.2]{VartziotisHimpel2013} we have been content with results about the dynamical behavior for volume elements. If we look more closely, our proof implicitly shows something else: on the mesh given by the single volume element, the GETMe algorithm with our normalization is optimization-based. The same is true above. The vector fields $Y$ and $\tilde Y$ yield optimization-based GETMe transformations.

\subsection{Global Optimization-based algorithm}

Since there are many quality functions which might yield global optimization-based algorithms, let us consider an arbitrary homogeneous degree $d$ vector field $X$ and a quality measure $q$ satisfying 
\[
  \D(\pi)(X_x) = c_x\nabla_\pi q|_x  \quad \text{and} \quad c_x>0 \quad \text{for }x\in N\;.
 \]
 In particular the following theorem applies to $q_i$ in Sect.\ \ref{sec_concrete_quality} and similar algebraic quality measures.

\begin{theorem}\label{thm_main}
 The GETMe operator given by
 \begin{align*}
   T_\sigma\co N &\to N\;,\\
   x&\mapsto x' = \pi(x + \sigma \, \Psi(X)) \quad \text{for } X=\sum_{e\in E} X_e\;.
 \end{align*}
 is a global optimization-based smoothing method for mixed-volume meshes as long as $\sigma>0$ sufficiently small, which is invariant under translation and scaling. More precisely, for the mesh quality function $q$ and for any $x \in N$ we have 
 \[
 q(T_\sigma(x))>q(x) \quad \text{for }\sigma>0 \text{ sufficiently small.}
 \]
\end{theorem}

\begin{proof}
Fix $x \in N$, and let $x_\sigma = x + \sigma\,  \Psi(X) \in \R^{3\times n}$. Since $\pi(x) = x$ and
\[
\frac{d}{d\sigma}(x_\sigma) = \Psi(X) = \|X\|^{(d-1)/d} X\;,
\]
the chain rule gives
\[
\left.\frac{d}{d\sigma}q(T_\sigma(x))\right|_{\sigma=0}  = \left. \frac{d}{d\sigma}(q(\pi(x_\sigma)))\right|_{\sigma=0} =  \|X\|^{(d-1)/d} \D q_x \circ \D\pi_x (X_x)\;.
\]
Now consider the flow line $\gamma$ of $\D\pi(X)$ on $N$ starting at $x$. The path $\gamma(t)\in N$ therefore satisfies the initial value problem
\[\gamma(0)=x \quad\text{and}\quad  \dot\gamma(t)=\D\pi (X_{\gamma(t)})\;. \]
Then
\[
 \left.\frac{d}{d\sigma}q(T_\sigma(x))\right|_{\sigma=0}  = \|X_x\|^{(d-1)/d} \D q ( \dot\gamma(0)) = \|X_x\|^{(d-1)/d} \frac{d}{dt}(q\circ\gamma(t))|_{t=0}\;.
\]
Fig.\ \ref{fig_proof} shows the flow line $\gamma$, the linear path in the direction of the vector field $\Psi(X)$ and its projection to the sphere $N$ via $\pi$.
\begin{figure}[ht]
\def\svgwidth{4cm}
\centering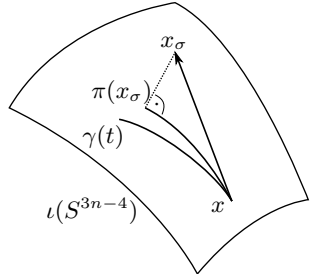\caption{The different players in the proof\label{fig_proof}}
\end{figure}
Since by assumption we have $\D\pi( X_x) = c_x\nabla_\pi(q|_{N})_x\;,$ we get
\begin{align*}
\|X_x\|^{(1-d)/d}&\cdot\left.\frac{d}{d\sigma}q(T_\sigma(x))\right|_{\sigma=0}  = \left.\frac{d}{dt}(q \circ \gamma(t))\right|_{t=0} \\
& =  \langle \nabla_\pi(q|_{N})_x, \dot\gamma(0)\rangle_{\iota}
 = \langle c_x^{-1}\D\pi (X_{\gamma(0)}), \dot\gamma(0)\rangle_{\iota} = c_x^{-1}\| \dot \gamma(0) \|_{\iota}^2 >0\;.
\end{align*}
Given $x\in N$, we therefore have that
\[
q(T_\sigma(x))>q(x) \quad \text{for }\sigma>0 \text{ sufficiently small.}\rqedhere
\]
\end{proof}

\section{Generalizations}\label{sec_generalizations}

As we have discussed, the GETMe algorithm for volume meshes was originally conceived for tetrahedral meshes \cite{VartziotisWipperSchwald2009} and later generalized to mixed volume meshes \cite{VartziotisWipper2011Hex,VartziotisWipper2011Mixed} via dual polyhedra. Furthermore, it is standard in applications to modify smoothing algorithms, so that they preserve the shape or features like corners and edges of a model. Let us outline, how such generalizations fit into our mathematical framework. Lastly, we have focused only on volume meshes, but with the same approach we can generalize the smoothing methods to arbitrary dimensions, in particular surface meshes. This allows us to construct a GETMe algorithm optimizing the isoperimetric quotient as a quality metric.

\subsection{Mixed Volume Meshes}\label{sec_mixed}

The most poignant feature of our systematic approach is the natural compatibility of the GETMe algorithm for different volume elements by way of the mean volume. Let us give an explicit description of a method globally optimizing this mixed mesh quality. As we have already seen in \eqref{eq_getme_tetra}, the vector $X_e$ for tetrahedra $e$ has an elegant description in terms of face normals. Not surprisingly, the same is true for hexahedra, prisms and pyramids. We compute the following vectors
\begin{align*}
X_e &= 
\frac{1}{2}\begin{bmatrix}
\nu(5,4,2) + \nu(5,4,3,2)\\
\nu(5,1,3) + \nu(5,1,4,3)\\
\nu(5,2,4) + \nu(5,2,1,4)\\
\nu(5,3,1) + \nu(5,3,2,1)\\
2 \cdot \nu(1,2,3,4)\end{bmatrix}\text{ for the pyramid }e=\text{\raisebox{-1cm}{\def\svgwidth{2.5cm}
\begingroup%
  \makeatletter%
  \providecommand\color[2][]{%
    \errmessage{(Inkscape) Color is used for the text in Inkscape, but the package 'color.sty' is not loaded}%
    \renewcommand\color[2][]{}%
  }%
  \providecommand\transparent[1]{%
    \errmessage{(Inkscape) Transparency is used (non-zero) for the text in Inkscape, but the package 'transparent.sty' is not loaded}%
    \renewcommand\transparent[1]{}%
  }%
  \providecommand\rotatebox[2]{#2}%
  \ifx\svgwidth\undefined%
    \setlength{\unitlength}{94.82192383bp}%
    \ifx\svgscale\undefined%
      \relax%
    \else%
      \setlength{\unitlength}{\unitlength * \real{\svgscale}}%
    \fi%
  \else%
    \setlength{\unitlength}{\svgwidth}%
  \fi%
  \global\let\svgwidth\undefined%
  \global\let\svgscale\undefined%
  \makeatother%
  \begin{picture}(1,0.85190115)%
    \put(0,0){\includegraphics[width=\unitlength]{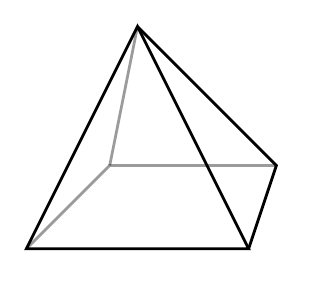}}%
    \put(-0.00420196,0.0119385){\color[rgb]{0,0,0}\makebox(0,0)[lb]{\smash{$v_1$}}}%
    \put(0.71293178,0.0119385){\color[rgb]{0,0,0}\makebox(0,0)[lb]{\smash{$v_2$}}}%
    \put(0.85226419,0.29444945){\color[rgb]{0,0,0}\makebox(0,0)[lb]{\smash{$v_3$}}}%
    \put(0.32218905,0.25865482){\color[rgb]{0,0,0}\makebox(0,0)[lb]{\smash{$\cg v_4$}}}%
    \put(0.37545708,0.8134409){\color[rgb]{0,0,0}\makebox(0,0)[lb]{\smash{$v_5$}}}%
  \end{picture}%
\endgroup%
},}\\
X_e &= \displaystyle\frac{1}{2}\begin{bmatrix}
\nu(3,2,4) + \nu(2,5,4,6,3)\\
\nu(1,3,5) + \nu(3,6,5,4,1)\\
\nu(2,1,6) + \nu(1,4,6,5,2)\\
\nu(5,6,1) + \nu(6,3,1,2,5)\\
\nu(6,4,2) + \nu(4,1,2,3,6)\\
\nu(4,5,3) + \nu(5,2,3,1,4)\\
\end{bmatrix}\text{ for the prism }e=\text{\raisebox{-1.3cm}{\def\svgwidth{2.5cm}
\begingroup%
  \makeatletter%
  \providecommand\color[2][]{%
    \errmessage{(Inkscape) Color is used for the text in Inkscape, but the package 'color.sty' is not loaded}%
    \renewcommand\color[2][]{}%
  }%
  \providecommand\transparent[1]{%
    \errmessage{(Inkscape) Transparency is used (non-zero) for the text in Inkscape, but the package 'transparent.sty' is not loaded}%
    \renewcommand\transparent[1]{}%
  }%
  \providecommand\rotatebox[2]{#2}%
  \ifx\svgwidth\undefined%
    \setlength{\unitlength}{91.47348633bp}%
    \ifx\svgscale\undefined%
      \relax%
    \else%
      \setlength{\unitlength}{\unitlength * \real{\svgscale}}%
    \fi%
  \else%
    \setlength{\unitlength}{\svgwidth}%
  \fi%
  \global\let\svgwidth\undefined%
  \global\let\svgscale\undefined%
  \makeatother%
  \begin{picture}(1,1.07121017)%
    \put(0,0){\includegraphics[width=\unitlength]{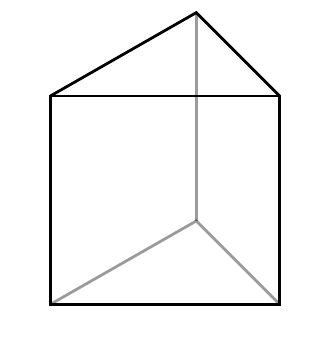}}%
    \put(0.02213596,0.05502523){\color[rgb]{0,0,0}\makebox(0,0)[lb]{\smash{$v_1$}}}%
    \put(0.83697969,0.00825034){\color[rgb]{0,0,0}\makebox(0,0)[lb]{\smash{$v_2$}}}%
    \put(0.47336148,0.39745657){\color[rgb]{0,0,0}\makebox(0,0)[lb]{\smash{$\cg v_3$}}}%
    \put(-0.00290385,0.7691415){\color[rgb]{0,0,0}\makebox(0,0)[lb]{\smash{$v_4$}}}%
    \put(0.89790416,0.76039526){\color[rgb]{0,0,0}\makebox(0,0)[lb]{\smash{$v_5$}}}%
    \put(0.46499156,1.04463143){\color[rgb]{0,0,0}\makebox(0,0)[lb]{\smash{$v_6$}}}%
  \end{picture}%
\endgroup%
},}\\
X_e &= \displaystyle\frac{1}{2}\begin{bmatrix}
 \nu(2,5,4) + \nu(6,5,8,4,3,2)\\
 \nu(3,6,1) + \nu(7,6,5,1,4,3)\\
 \nu(4,7,2) + \nu(8,7,6,2,1,4) \\
 \nu(1,8,3) + \nu(5,8,7,3,2,1) \\
 \nu(1,6,8) + \nu(6,7,8,4,1,2)\\
 \nu(2,7,5) + \nu(7,8,5,1,2,3)\\
 \nu(3,8,6) + \nu(8,5,6,2,3,4)\\
 \nu(4,5,7) + \nu(5,6,7,3,4,1)
\end{bmatrix} \text{ for the hexahedron }e=\text{\raisebox{-1cm}{\def\svgwidth{1.8cm}
\begingroup%
  \makeatletter%
  \providecommand\color[2][]{%
    \errmessage{(Inkscape) Color is used for the text in Inkscape, but the package 'color.sty' is not loaded}%
    \renewcommand\color[2][]{}%
  }%
  \providecommand\transparent[1]{%
    \errmessage{(Inkscape) Transparency is used (non-zero) for the text in Inkscape, but the package 'transparent.sty' is not loaded}%
    \renewcommand\transparent[1]{}%
  }%
  \providecommand\rotatebox[2]{#2}%
  \ifx\svgwidth\undefined%
    \setlength{\unitlength}{82.51834717bp}%
    \ifx\svgscale\undefined%
      \relax%
    \else%
      \setlength{\unitlength}{\unitlength * \real{\svgscale}}%
    \fi%
  \else%
    \setlength{\unitlength}{\svgwidth}%
  \fi%
  \global\let\svgwidth\undefined%
  \global\let\svgscale\undefined%
  \makeatother%
  \begin{picture}(1,1.12062651)%
    \put(0,0){\includegraphics[width=\unitlength]{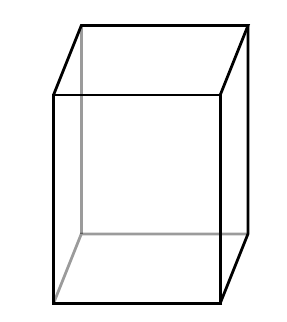}}%
    \put(0.01834893,0.02046361){\color[rgb]{0,0,0}\makebox(0,0)[lb]{\smash{$v_1$}}}%
    \put(0.80653243,0.00914569){\color[rgb]{0,0,0}\makebox(0,0)[lb]{\smash{$v_2$}}}%
    \put(0.88682441,0.24724497){\color[rgb]{0,0,0}\makebox(0,0)[lb]{\smash{$v_3$}}}%
    \put(0.30513589,0.18318221){\color[rgb]{0,0,0}\makebox(0,0)[lb]{\smash{$\cg v_4$}}}%
    \put(-0.00321898,0.78473049){\color[rgb]{0,0,0}\makebox(0,0)[lb]{\smash{$v_5$}}}%
    \put(0.57462585,0.67646474){\color[rgb]{0,0,0}\makebox(0,0)[lb]{\smash{$v_6$}}}%
    \put(0.86547026,1.0313716){\color[rgb]{0,0,0}\makebox(0,0)[lb]{\smash{$v_7$}}}%
    \put(0.21797423,1.09116337){\color[rgb]{0,0,0}\makebox(0,0)[lb]{\smash{$v_8$}}}%
  \end{picture}%
\endgroup%
}.}
\end{align*}
For details we refer to \cite[Sects.\ 5 and 6]{VartziotisHimpel2013}.

\subsection{Isoperimetric quotient}\label{sec_iq}

If we want to construct a quality measure, which measures how round a given polyhedron is, we can try the isoperimetric quotient given by $36\pi$ times the volume squared divided by the surface area of the boundary cubed, where the constant normalizes the measure to be 1 for the unit sphere. Steiner \cite{steiner1841} conjectured that platonic solids maximize the isoperimetric quotient, which has been confirmed for all polyhedra but the icosahedron \cite{steinitz1927,fejes,bezdek2007}. Instead of the isoperimetric quotient we will consider its root, which allows us to naturally introduce a sign using the signed volume \eqref{eq_vol}. This generalizes to a function on the space of mixed volume meshes
\[
\iq(x_e) = 6\sqrt{\pi} \frac{\vol(x_e)}{\area(x_e)^{3/2}},
\]
where $\area$ is the mean surface area of the boundary of the polyhedron mesh $p$ defined analogously to the mean volume $\vol$. The function $\iq$ extends linearly (or by some other combining function $Q$) to all polyhedral meshes. A straight-forward computation of the gradient of $\iq$ gives us yet another GETMe global optimization-based smoothing algorithm. We expect this algorithm to be the most powerful, and by the very definition of the quality measure, it yields a scaling-invariant transformation.

As a side note, all tested meshes with an icosahedron structure converge to the regular icosahedron. A proof of this heuristic observation will not be easy and would imply Steiner's conjecture. Note that the isoperimetric volume element quality measure can be constructed using the mean volume function restricted to all polyhedra with a fixed boundary surface area, which 
exemplifies that different embeddings of the sphere of polyhedra induce different quality measures from the same volume function.

\subsection{Arbitrary dimensions}

The gradient of the volume has a particular nice and computationally efficient expression. Nevertheless, we get similar transformations in arbitrary dimensions. In particular, the area of a polygons is the two-dimensional analogue of the volume for polyhedra, and we also have an isoperimetric quotient for polygons. Therefore we can construct a GETMe transformation for triangulations of the plane and even of surface meshes. These will yield global optimization based smoothing methods for surface meshes. We have performed preliminary tests for the two-dimensional analogue of the measures in \eqref{eq_vol_measures}, which show that the resulting transformation based on area is a powerful GETMe smoothing for triangular meshes.

\subsection{Shape Preservation}

In order to preserve shape or features of a model, there exist standard projection techniques. If we consider the flow of $X$ rather than the discrete transformation method, we can constrain feature vertices to the submanifolds of $N$ given by the boundary surface, edges or corners. These submanifolds might not be smooth. Nevertheless, this approach fits nicely into the mathematical framework by considering the restriction of the volume function to these submanifolds of $N$.

\section{Conclusions}

We have described the mathematical framework and a systematic approach to global optimization-based GETMe smoothing methods for mixed volume meshes, which enables us to systematically prove, generalize and analyze properties of GETMe. These methods have the potential of being both runtime efficient and provably quality effective. We have given explicit constructions of potential smoothing methods, which will be analyzed numerically and theoretically in future publications.
\bibliographystyle{spmpsci}\bibliography{literature.bib}\end{document}